\documentclass{amsart}

\usepackage{amsmath,amsthm,amssymb,cite}

\newtheorem{theorem}{Theorem}[section]
\newtheorem{corollary}[theorem]{Corollary}
\newtheorem{lemma}[theorem]{Lemma}
\newtheorem{proposition}[theorem]{Proposition}
\newtheorem{remark}[theorem]{Remark}

\numberwithin{equation}{section}

\allowdisplaybreaks[3]

\newcommand\R{{\mathbb R}}
\newcommand\C{{\mathbb C}}
\newcommand\X{{\R^d}}
\newcommand\N{{\mathbb N}}
\renewcommand\L{{\mathcal L}}
\newcommand\F{{\mathcal F}}
\newcommand\K{{\mathcal K}}
\newcommand\M{{\mathcal M}}
\newcommand\B{{\mathcal B}}
\newcommand\Bc{\B_{\mathrm{b}}}

\newcommand\Ga{\Gamma}
\newcommand\ga{\gamma}
\newcommand\La{\Lambda}
\newcommand\la{\lambda}
\newcommand\eps{\varepsilon}
\newcommand\hL{\hat{L}}
\renewcommand\a{\alpha}
\newcommand\aC{{\a C}}

\begin{document}

\title[Functional evolutions for death-immigration dynamics]{Functional
evolutions for homogeneous stationary death-immigration
spatial dynamics}

\author{D. Finkelshtein}
\address{Institute of Mathematics, National Academy of Sciences of
Ukraine, Kyiv, Ukraine}
\curraddr{}
\email{fdl@imath.kiev.ua}
\thanks{Author would like to thank Prof. Dr. Yuri Kondratiev for useful
discussions. This work was supported by DFG through SFB-701,
Bielefeld University, Germany.}

\date{DD/MM/2011}

\begin{abstract}
We discover death-immigration non-equilibrium stochastic dynamics in
the continuum also known as the Surgailis process. Explicit
expression for the correlation functions is presented. Dynamics of
states and their generating functionals are studied. Ergodic
properties for the evolutions are considered.
\end{abstract}

\maketitle

\section{Introduction}

Complex systems theory is a growing interdisciplinary area with a
very broad spectrum of motivations and applications. One may
characterize complex systems by properties as diversity and
individuality of components, localization of interactions among
components, and the outcomes of interactions used for replication or
enhancement of components. In the study of these systems,  proper
language and techniques are delivered by the interacting particle
models which form a rich and powerful direction in modern stochastic
and infinite dimensional analysis.  Interacting particle systems
have  a wide use as models in condensed matter physics, chemical
kinetics, population biology, ecology, sociology and economics.

Mathematical realizations of such models may be considered as a
dynamics of collections of points in proper spaces. The possible
positions of points may be fixed due to the structure of space, e.g.
dynamics on graphs, or, in particular, on lattices. Another area of
models connects with free positions of points in the continuum, say,
in the Euclidean space $\X$. However, as was shown in statistical
physics, many empirical effects, such as phase transitions, are
impossible in systems with a finite number of points. Due to this,
one can consider infinite point systems as a mathematical
approximation for realistic systems with a high number of elements.
The connection with the reality, where infinite systems are absent,
is given by the restriction of the study to locally finite systems
(configurations) which have only finite number of elements in any
finite volume.

Depending on applications, the points of such a system may be
interpreted as molecules in physics, plants in ecology, animals in
biology, infected people in medicine, companies in economics, market
agents in finance, and so on. For study stochastic dynamics of such
systems we may consider different mechanisms of (random) evolutions
of their points. Existing points may disappear from the
configuration that is naturally called `death'. Each existing point
may change own position due to some moving or hop; this mechanism
traditionally is called `emigration'. Each existing point may
produce a new one, that is called `birth'. There exists also another
possibility for appearing a new element in the configuration coming
from outside; this is called `immigration'. Mathematically, the
random evolution of the system is described by a heuristic Markov
generator which includes parts corresponding to different mechanisms
above.

Rigorous mathematical results concerning stochastic dynamics of
configurations in the continuum have not very reach history. One of
the pioneering work in this area was \cite{HS1978}. Special class of
models introduced therein have been recently studied in
\cite{GK2006, GK2008}. We mentioned also
\cite{Pre1975,Pen2008,Qi2008}, and references therein. During the
last decade a functional approach for studying of the stochastic
dynamics above was discovered. It was considered the evolutional
equations connected with considered stochastic dynamics, namely,
equations
 on states of systems  and their correlation functions, equations on generating functionals
and so on. Studying this evolutional equations yields not only
existence (in different senses) of dynamics but their qualitative
and quantitative properties also. For general description of this
approach see, e.g., \cite{FKO2009,KKM2008}, and for particular
models see, e.g.,
\cite{KKZ2006,FKKZ2010,FKK2010,FKK2009,Finkelshtein2009}.

In the present paper we consider one of the simplest model, where
only independent (constant) death and immigration appear. The
corresponding stochastic process is the well-known Surgailis process
\cite{Sur1983,Sur1984,KLR2008}. For this model we find explicit
expression for correlation functions that gives us a way to improve
general results as well as to obtain new ones. The structure of the
paper is the following. We describe the model and present necessary
knowledge on configuration space techniques in Section 2. Section 3
is devoted to the evolutions of correlation functions and measures
(states)\ of the system. The~ergodic properties of the dynamics as
well as evolution of the generating functionals are presented in
Section 4. Finally, Section 5 deals with the so-called dynamics of
quasi-observables.

We also note that the main results obtained in this work may be
generalized to death and immigration rates whose are independent of
other points of a configuration, however, they may depend on the
position of each point and time. We will consider this case of
non-homogeneous non-stationary death-immigration process in a
forthcoming publication.

\section{Description of model}

The simplest economic model in  the description of spatial dynamics is the model of free
development when particles (which we may
interpret, for instance, as companies on the
market) appears independently without any influence of existing ones. On the other language,
they migrate from the outside without any motivation due to
situation inside the system.
Of course, companies on real market never have infinite life time. We consider model with global regulation. This means
that any points of configuration has exponentially distributed (with
some positive parameter $m$) random life time. Hence, again a death
(bankruptcy) appears due to ``request'' from the outside.

The state space of this model is the space $\Ga=\Ga_\X$ of all locally finite subsets
(configurations) in $\X$:
\[
\Ga :=\left\{  \ga \subset \X \bigm| |\ga_\La |<\infty, \text{ for
all } \La \in \B_c(\X)\right\}.
\]
Here $\ga_\La=\ga\cap\La$, $|\cdot|$ means cardinality of a set,
$\B_{c}(\X)$ denote the system of all bounded Borel sets in $\X$. We
consider the $\sigma $-algebra $ \B(\Ga )$ as the smallest $\sigma
$-algebra for which all the mappings $N_\La :\Ga \rightarrow
\N_0:=\N\cup\{0\}$, $N_\La (\ga ):=|\ga _ \La |$ are measurable for
all $\La\in\B_c(\X)$. For every $\La \in \B_c(\X)$ one can define a
projection $p_\La :\Ga \rightarrow \Ga _\La :=\{\ga\in\Ga\mid
\ga\subset\La\}$; $p_\La (\ga ):=\ga _\La $ and w.r.t. this
projections $\Ga $ is the projective limit of the spaces $\{\Ga _\La
\}_{\La \in \B_c(\X)}$. One can consider also the $\sigma $-algebra
$ \B_\La(\Ga )$ as the smallest $\sigma $-algebra for which all the
mappings $N_{\La'} :\Ga \rightarrow \N_0$ are measurable for all
$\La'\in\B_c(\X)$, $\La'\subset\La$.

On $\Ga $ we consider the set of a cylinder functions
$\mathcal{F}L^0(\Ga )$, i.e. the set of all measurable function $F$
on $\bigl(\Ga,\B(\Ga)\bigr)$ which are measurable w.r.t. $\B_\La
(\Ga )$ for some $\La \in \B_c(\X)$. These functions are
characterized by the following relation: $F(\ga )=F\upharpoonright
_{\Ga _\La }(\ga _\La )$.

Let $ \mathcal{M}_{\mathrm{fm}}^1(\Ga )$ be the set of all
probability measures $\mu $ on $\bigl( \Ga, \B(\Ga) \bigr)$ which
have finite local moments of all orders, i.e. $\int_\Ga |\ga _\La
|^n\mu (d\ga )<+\infty $ for all $\La \in \B_c(\X)$ and $n\in \N_0$.

To describe a (pre-)generator of a dynamics above we consider for fixed $m>0$, $\sigma\geq0$ and for any
$F\in\mathcal{F}L^0(\Ga )$ the following expression
\begin{equation}\label{SurgGen}
\left( L F\right) \left( \gamma \right)
=m\sum_{x\in\ga}\bigl[F(\ga\setminus x ) - F(\ga) \bigr] + \sigma
\int_{\mathbb{R} ^{d}}\left[ F\left( \gamma \cup x\right) -F\left(
\gamma \right) \right] dx,
\end{equation}
which is well-defined since, by the definition of
$\mathcal{F}L^0(\Ga )$,  there exists $\La\in\Bc(\X)$ such that
$F(\gamma\setminus x)=F(\gamma)$ for any $x\in\gamma_{\Lambda^c}$
and $F(\gamma\cup x)=F(\gamma)$ for any $x\in\Lambda^c$; therefore,
sum and integral in \eqref{SurgGen} are finite. Stress that $L$ is
the generator of the (non-equilibrium) Surgailis process, see
\cite{Sur1983,Sur1984,KLR2008}.

We consider now the space of finite configurations on $\X$. The space of $n$-point configuration is
\[
\Ga ^{(n)}:=\left\{ \left. \eta \subset \X\right| \,|\eta
|=n\right\} ,\quad n\in \N_0.
\]
As a set,
$\Ga^{(n)}$ is equivalent to the symmetrization of
\[
\widetilde{(\X)^n} = \left\{ \left. (x_1,\ldots ,x_n)\in
(\X)^n\right| \,x_k\neq x_l\,\,\mathrm{if} \,\,k\neq l\right\}.
\]
Hence, $\Ga_{0}^{(n)}$ inherits the structure of an $n\cdot
d$-dimensional manifold. Applying this we can define Borel $\sigma
$-algebra $\B(\Ga _{0}^{(n)})$. Also one can consider a measure
$\mathrm{m}^{(n)}$ as image of product $\mathrm{m}^{\otimes n}$ of
Lebesgue measures $d\mathrm{m}(x)=dx$ on $\bigl(\X, \B(\X)\bigr)$.

The space of finite configurations
\[
\Ga _{0}:=\bigsqcup_{n\in \N_0}\Ga _{0}^{(n)}
\]
has structure of disjoint union, therefore, one can define the Borel
$\sigma $-algebra $\B (\Ga _0)$. A set $B\in \B (\Ga _0)$ is called
bounded if there exists a $\La \in \B_c(\X)$ and an $N\in \N$ such
that $B\subset \bigsqcup_{n=0}^N\Ga _\La ^{(n)}$, where
$\Ga_\La^{(n)}:=\bigl\{ \eta \subset \La\bigm| |\eta |=n\bigr\}$.

We will use also the following two classes of functions on $\Ga_0$:
$L_{\mathrm{ls}}^0(\Ga _0)$ is the set of all measurable functions
on $\Ga_0$ which have a local support, i.e. $G\in
L_{\mathrm{ls}}^0(\Ga _0)$ if there exists $\La \in \B_c(\X)$ such
that $G\upharpoonright_{\Ga _0\setminus \Ga _\La }=0$;
$B_{\mathrm{bs}}(\Ga _0)$ is the set of bounded measurable functions
with bounded support: $G\upharpoonright_{\Ga _0\setminus B}=0$ for
some bounded $B\in \B(\Ga_0)$.

The
Lebesgue---Poisson measure $\la_{z} $ on $\bigl(\Ga_0, \B(\Ga_0)\bigr)$ is defined as
\begin{equation}
\la _{z} :=\sum_{n=0}^\infty \frac {z^{n}}{n!}\mathrm{m} ^{(n)}.\label{defLPm}
\end{equation}
Here $z>0$ is the so called activity parameter. The restriction of
$\la _{z} $ to $\Ga _\La $ will be also denoted by $\la _{z} $. Let
$\la$ be the Lebesgue-Poisson measure on $\Ga_{0}$ (and $\Ga_\La$)
with activity parameter equal to 1.

The Poisson measure $\pi _{z} $ on $\bigl(\Ga ,\B(\Ga )\bigr)$ is
given as the projective limit of the family of measures $\{\pi _{z}
^\La \}_{\La \in \B_c(\X)}$, where $\pi _{z} ^\La $ is the measure
on $\Ga _\La $ defined by $\pi _{z} ^\La :=e^{-z \mathrm{m} (\La
)}\la _{z}$. Again, we will omit index in the case $z=1$.

The following mapping between functions on $\Ga _0$, e.g.
$L_{\mathrm{ls}}^0(\Ga _0)$, and functions on $\Ga $, e.g.
$\mathcal{F}L^{0}(\Ga )$, plays an  important role in our further
considerations:
\begin{equation}
KG(\ga ):=\sum_{\eta \Subset \ga }G(\eta ), \quad \ga \in \Ga,
\label{KT3.15}
\end{equation}
where $G\in L_{\mathrm{ls}}^0(\Ga _0)$, see, e.g.,
\cite{Len1975,Len1975a,KK2002}. The summation in the latter
expression is extend over all finite subconfigurations of $\ga ,$ in
symbols $\eta \Subset \ga $. The mapping $K$ is linear, positivity
preserving, and invertible, with
\begin{equation}
K^{-1}F(\eta ):=\sum_{\xi \subset \eta }(-1)^{|\eta \setminus \xi
|}F(\xi ),\quad \eta \in \Ga _0.\label{k-1trans}
\end{equation}

We consider now a mapping $\hL G:=K^{-1}L KG$ which is well-defined
on functions $G\in L_{\mathrm{ls}}^0(\Ga _0)$. By, e.g.,
\cite{FKO2009}, we have
\begin{equation}
\bigl( \hL G\bigr) \left( \eta \right) =-m|\eta| G(\eta)+\sigma
\int_{\mathbb{R}^{d}}G\left( \eta \cup x\right) dx\label{exprLP}.
\end{equation}

Let now $C>1$ be fixed. Applying results from \cite{FKKZ2010} to the
zero-potential case, we obtain that \eqref{exprLP} provides a linear
operator on the Banach space of ${\mathcal{B}}
(\Gamma_0)$-measurable functions
\begin{equation}  \label{norm}
\L _C:=\biggl\{ G:\Gamma_0\rightarrow{\mathbb{R}} \biggm| \|G\|_C:=
\int_{\Gamma_0} |G(\eta)| C^{|\eta|} d\lambda(\eta) <\infty\biggr\}
\end{equation}
with dense domain $\L_{2C}\subset\L _C$. If additionally,
\begin{equation}  \label{verysmallparam}
C\geq \frac{\sigma}{m}
\end{equation}
then $\bigl( \hL , \L_{2C}\bigr)$ is closable linear operator in
$\L_C$ and its closure $\bigl( \hL, D(\hL)\bigr)$ generates a
strongly continuous contraction semigroup ${\hat{T}} (t)$ on $\L_C$.

\section{Correlation functions evolution}

\subsection{Notion of correlation functions}
A measure $\rho $ on $\bigl( \Ga_0, \B(\Ga_0) \bigr)$ is
called locally finite iff $\rho (A)<\infty $ for all bounded sets
$A$ from $\B(\Ga _0)$, the set of such measures is denoted by
$\mathcal{M}_{\mathrm{lf}}(\Ga _0)$. One can define a transform
$K^{*}:\mathcal{M}_{\mathrm{fm}}^1(\Ga )\rightarrow
\mathcal{M}_{\mathrm{lf}}(\Ga _0),$ which is dual to the
$K$-transform, i.e., for every $\mu \in
\mathcal{M}_{\mathrm{fm}}^1(\Ga )$, $G\in \B_{\mathrm{bs}}(\Ga _0)$
we have
\[
\int_\Ga KG(\ga )\mu (d\ga )=\int_{\Ga _0}G(\eta )\,(K^{*}\mu
)(d\eta ).
\]
$\rho _\mu :=K^{*}\mu $ we call the correlation measure
corresponding to $\mu $.

As shown in \cite{KK2002} for $\mu \in
\mathcal{M}_{\mathrm{fm}}^1(\Ga )$ and any $G\in L^1(\Ga _0,\rho
_\mu )$ the series \eqref{KT3.15} is $\mu $-a.s. absolutely
convergent. Furthermore, $KG\in L^1(\Ga ,\mu )$ and
\begin{equation}
\int_{\Ga _0}G(\eta )\,\rho _\mu (d\eta )=\int_\Ga (KG)(\ga )\,\mu
(d\ga ). \label{Ktransform}
\end{equation}

Among the elements in the domain of the $K$-transform are also the
so-called coherent states $e_\lambda(f)$. By definition, for any
$\B(\X)$-measurable function $f$,
\[
e_\lambda (f,\eta ):=\prod_{x\in \eta }f(x) ,\ \eta \in \Gamma
_0\!\setminus\!\{\emptyset\},\quad  e_\lambda (f,\emptyset ):=1.
\]
Then, by \eqref{defLPm}, for $f\in L^1({{\mathbb{R}}^d},dx)$ we
obtain $e_\lambda(f)\in L^1(\Gamma_0,d\lambda)$ and
\begin{equation}  \label{LP-exp-mean}
\int_{\Gamma_0}
e_\lambda(f,\eta)d\lambda(\eta)=\exp\{\langle f\rangle\},
\end{equation}
here and below $\displaystyle\langle f\rangle=\int_{{\mathbb{R}}^d}
f(x)dx$.

Note that
\begin{equation}
\bigl( Ke_\lambda (f)\bigr) (\gamma) = \prod_{x\in \gamma
}\bigl(1+f(x)\bigr),\quad \mu
\mathrm{-a.a.}\,\gamma\in\Gamma,\label{kexp}
\end{equation}
for all $\B (\X)$-measurable functions $f$ such that
$e_\lambda(f)\in L^1(\Gamma_0,\rho_\mu)$, see,~e.g., \cite{KK2002}.

Let $\mu \in \mathcal{M}_{\mathrm{fm} }^1(\Ga )$. If for all for all
$\La \in \B_\La (\X)$ the projection
 $\mu_\La :=\mu
\circ p_\La ^{-1}$ is absolutely continuous with respect to (w.r.t.) $\pi^\La$ on
$\Ga_\La$ then   $\rho
_\mu :=K^{*}\mu $ is absolutely continuous w.r.t. $\la $ on $\Ga_0$. The corresponding Radon--Nikodym derivative
\[
k_{\mu}(\eta):=\frac{d\rho_{\mu}}{d\la}(\eta),\quad
\eta\in\Ga_{0}
\]
is called a correlation functional of a measure $\mu$.
The functions
\begin{equation}\label{cf}
k_{\mu}^{(n)}:(\X)^{n}\longrightarrow\R_{+},
\end{equation}
given by
\[ k_{\mu}^{(n)}(x_{1},\ldots,x_{n}):=\left\{\begin{array}{ll}
k_{\mu}(\{x_{1},\ldots,x_{n}\}), & \mbox{if $(x_{1},\ldots,x_{n})\in
\widetilde{(\X)^{n}}$}\\ 0, & \mbox{otherwise}
\end{array},
\right.\] are well known correlation functions of statistical
physics, see e.g \cite{Rue1970,Rue1969}.

Obviously, not any positive function on $\Ga_0$ is a correlation functional of a some
measure on $\Ga$. To describe sufficient condition on this we will do in
the following manner. Given $G_1$ and $G_2$ two $\mathcal{B}(\Gamma_0)$-measurable
functions, let us consider the $\star$-convolution between $G_1$ and
$G_2$,
\begin{equation}\label{rsdef}
(G_1\star G_2)(\eta ):=\sum_{\xi_1 \sqcup\xi_{2}\sqcup \xi_{3}=\eta
} G_1(\eta _1\cup\eta _2)G_2(\eta _2\cup\eta _3),
\end{equation}
where sign $\sqcup$ denotes disjoint union (parts may be empty), see
\cite{KK2002} for a details. It is straightforward to verify that
the space of all $\mathcal{B}(\Gamma_0)$-measurable functions
endowed with this product has the structure of a commutative algebra
with unit element $e_\lambda(0)$. Furthermore, for every $G_1,
G_2\in B_{bs}(\Gamma_0)$ we have $G_1\star G_2\in B_{bs}(\Gamma_0)$,
and
\begin{equation}
K\left( G_1\star G_2\right) =\left( KG_1\right) \cdot \left(
KG_2\right) \label{likefourier}
\end{equation}
cf.~\cite{KK2002}. Note that
\begin{equation}
e_\lambda(f)\star e_\lambda(g) = e_\lambda(f + g+ fg) \label{rscoh}
\end{equation}
for all $\mathcal{B}(\mathbb{R}^d)$-measurable functions $f$ and
$g$.

The following theorem shows when we can reconstruct a measure $\mu
\in \mathcal{M}_{\mathrm{fm} }^1(\Ga )$ by the system of symmetric
functions \eqref{cf}.
\begin{theorem}[\cite{KK2002}]\label{posdef} Let $k:\Ga_0\rightarrow
\R_+$ be measurable function such that $kd\la \in
\M_{\mathrm{lf}}(\Ga_0)$, $k(\emptyset)=1$, there exists
$C>0,\eps>0$ such that  $k(\eta)\leq
C^{|\eta|}\bigl(|\eta|!\bigr)^{1-\eps}$, $\eta\in\Ga_0$ and the
function $k$ is positive definite in the sense that for any $G\in
B_{\mathrm{bs}}(\Ga_0)$
\begin{equation}\label{ex:posdef}
\int_{\Ga_0}(G\star\bar{G})(\eta) k(\eta) d\la(\eta)\geq0.
\end{equation}
Then there exists a unique measure $\mu\in\M_{\mathrm{fm}}^1(\Ga)$
such that $k=k_\mu$.
\end{theorem}

\subsection{Evolution of correlation functions}

 The space $(\L _C)^{\prime
}=\bigl(L^1(\Gamma_0, d\lambda_C)\bigr)^{\prime} =L^\infty(\Gamma_0,
d\lambda_C)$ is the
topologically dual space to the space $\L_C$. The space $L^\infty(\Gamma_0, d\lambda_C)$ is
isometrically isomorphic to the Banach space
\begin{equation*}
{\mathcal{K}}_{C}:=\left\{k:\Gamma_{0}\rightarrow{\mathbb{R}}\,\Bigm| k\cdot
C^{-|\cdot|}\in L^{\infty}(\Gamma_{0},\lambda)\right\}
\end{equation*}
with the norm $
\|k\|_{{\mathcal{K}}_C}:=\|C^{-|\cdot|}k(\cdot)\|_{L^{\infty}(\Gamma_{0},
\lambda)}, $ where the isomorphism is provided by the isometry $R_C$
\begin{equation}  \label{isometry}
(\L _C)^{\prime }\ni k \longmapsto R_Ck:=k\cdot C^{|\cdot|}\in
{\mathcal{K}} _C.
\end{equation}

In fact, we may say about a duality between Banach spaces $\L _C$
and ${\mathcal{K}}_C$, which is given by the following expression
\begin{equation}
\left\langle\!\left\langle G,\,k \right\rangle\!\right\rangle :=
\int_{\Gamma_{0}}G\cdot k\, d\lambda,\quad G\in\L _C, \ k\in {\mathcal{K}}_C
\label{duality}
\end{equation}
with
\begin{equation}
\left\vert \left\langle\!\left\langle G,k \right\rangle\!\right\rangle
\right\vert \leq \|G\|_C \cdot\|k\|_{{\mathcal{K}}_C}.  \label{funct_est}
\end{equation}
It is clear that for any $k\in {\mathcal{K}}_C$
\begin{equation}  \label{RB-norm}
|k(\eta)|\leq \|k\|_{{\mathcal{K}}_C} \, C^{|\eta|} \quad \text{for
} \lambda \text{-a.a. } \eta\in\Gamma_0.
\end{equation}

Let $\bigl( {\hat{L}} ^{\prime }, D({\hat{L}} ^{\prime })\bigr)$ be
an operator in $(\L_C)^{\prime }$ which is dual to the closed
operator $\bigl( {\hat{L}} , D({\hat{L}} )\bigr)$. We consider also
its image in ${\mathcal{K}}_C$ under isometry $R_C$, namely, let
${\hat{L}}^{*}=R_C{\hat{L}} ^{\prime }R_{C^{-1}}$ with a domain
$D({\hat{L}} ^{*})=R_C  D({\hat{L} } ^{\prime })$. Then, for any
$G\in\L _C$, $k\in D({\hat{L}}^\ast)$
\begin{align*}
\int_{\Gamma_0}G\cdot {\hat{L}}^\ast k
d\lambda=&\int_{\Gamma_0}G\cdot R_C{ \hat{L}} ^{\prime }R_{C^{-1}} k
d\lambda=\int_{\Gamma_0}G\cdot {\hat{L}}
^{\prime }R_{C^{-1}} k d\lambda_C \\
=& \int_{\Gamma_0}{\hat{L}} G\cdot R_{C^{-1}} k
d\lambda_C=\int_{\Gamma_0}{ \hat{L}} G\cdot k d\lambda,
\end{align*}
therefore, ${\hat{L}}^\ast$ is the dual operator to ${\hat{L}} $ w.r.t.
duality \eqref{duality}.

By, e.g., \cite{FKO2009}, we have the precise form of
${\hat{L}}^{*}$ on $D({\hat{L}}^\ast)$:
\begin{equation}
\bigl( \hL^* k\bigr) \left( \eta \right) =-m|\eta| k(\eta)+\sigma
\sum_{x\in \eta }k\left( \eta \setminus x\right) .\label{exprLPdual}
\end{equation}

In the same way one can consider the adjoint contraction semigroup
${\hat{T}} ^{\prime }(t)$ in $(\L _C)^{\prime }$ and its image
${\hat{T}} ^\ast(t)$ in ${\mathcal{K}}_C$. Now, we may apply general
results about adjoint semigroups (see, e.g., \cite{EN2000}) onto the
contraction semigroup ${\hat{T}}^\ast(t)$. The last semigroup will
be weak*-continuous, moreover, weak*-differentiable at $0$ and
${\hat{L}}^\ast$ will be weak*-generator of ${\hat{T}}^\ast(t)$.
Here and below we mean ``weak*-properties'' w.r.t. duality
\eqref{duality}. Let $\mathring{\K}_C=\bigl\{ k\in{\mathcal{K}}_C
\bigm| \exists \lim_{t\downarrow0}\bigl\| {\hat{T}}^\ast(t)k -
k\bigr\|_{{\mathcal{K}}_C} =0\bigr\}$. Then $\mathring{\K}_C$ is
closed, weak*-dense, ${\hat{T}}^\ast(t)$-invariant linear subspace
of ${\mathcal{K}}_C$. Moreover, $\mathring{\K}_C=\overline{D(
{\hat{L}}^\ast)}$ (the closure is in the norm of ${\mathcal{K}}_C$).
Let ${ \hat{T}}^\odot(t)$ denote the restriction of
${\hat{T}}^\ast(t)$ onto Banach space ${\mathring{\K}}_C$. Then
${\hat{T}} ^\odot(t)$ is a contraction $C_0$-semigroup on
${\mathring{\K}}_C$ and its generator ${\hat{L}} ^\odot$ will be
part of $ {\hat{L}}^\ast$, namely, $D({\hat{L}} ^\odot)=\bigl\{k\in
D({\hat{L}}^\ast) \bigm| {\hat{L}}^\ast k\in
\overline{D({\hat{L}}^\ast)}\bigr\}$ and ${\hat{L} }^\ast k
={\hat{L}}^\odot k$ for any $k\in D({\hat{L}}^\odot)$.

Using simple reccurent structure of the operator \eqref{exprLPdual} we may
find explicit expression for the action of the contraction semigroup ${\hat{T}}^\ast(t)$
from the solution of the Cauchy problem
\begin{equation}
\frac{\partial }{\partial t}k_{t}=\hL^* k_{t} ,\qquad  k_t
\bigm|_{t=0}=k_0.  \label{KmE}
\end{equation}

To do this let us define the following associative and commutative convolution
on measurable functions on $\Ga_0$
\begin{equation}
(G_1 * G_2) (\eta) = \sum_{\xi\subset\eta} G_1(\xi) G_2(\eta
\setminus \xi), \quad \eta\in\Ga_0.\label{defast}
\end{equation}One can consider an algebra of measurable functions on $\Ga_0$
with such a product and the unit element $1^\ast (\eta):=0^{|\eta|}$. Note that,
\begin{align}
e_\la(f) * e_\la(g) =& e_\la(f+g) \label{addexp}\\
e_\la(f)\bigl( G_1 \ast G_2\bigr) =&\bigl(e_\la(f) G_1\bigr)\ast
\bigl(e_\la(f) G_2\bigr).
\end{align}

\begin{theorem}\label{solutionofCP}
The function
\begin{align}
k_{t}\left( \eta \right) = & \, e^{-tm|\eta|} \biggl( e_{\lambda
}\left(
 \frac{\sigma}{m}(e^{tm}-1)\right) \ast k_{0}\biggr) \left( \eta
\right) \label{freedev-m-dyn}\\ =& \biggl( e_{\lambda }\left(
\frac{\sigma}{m}(1-e^{-tm})\right) \ast \bigl( e_{\lambda }(
e^{-tm})  k_{0}\bigr)\biggr) \left( \eta \right).
\label{freedev-m-dyn-alt}
\end{align}
is a well-defined point-wise differentiable function which satisfied \eqref{KmE}.
\end{theorem}
\begin{proof}
By \eqref{exprLPdual}, \eqref{KmE} implies
\[
\frac{\partial }{\partial t}k_{t}^{(1)}(x_1)=-m k_{t}^{(1)}(x_1)
+\sigma,
\]
that yields
\[
k_{t}^{(1)}(x_{1})=e^{-mt}k_{0}^{(1)}(x_{1})+\sigma
\int_{0}^{t}e^{-m(t-s)}ds=e^{-mt}\left( k_{0}^{(1)}(x_{1})+\frac{
\sigma }{m}(e^{mt}-1)\right).
\]
Suppose that \eqref{freedev-m-dyn} holds for $|\eta|=n-1$, namely,
\[
k_{t}^{(n-1)}\left( x_{1},\ldots ,x_{n-1}\right)
=e^{-m(n-1)t}\sum_{\xi \subset \left\{ x_{1},\ldots ,x_{n-1}\right\}
}k_{0}^{(\left\vert \xi \right\vert )}\left( \xi \right) \left(
\frac{
\sigma }{m}(e^{mt}-1)\right) ^{n-1-\left\vert \xi \right\vert }.
\]
Then, by \eqref{exprLPdual} and \eqref{KmE} we obtain
\begin{align*}
&k_{t}^{(n)}\left( x_{1},\ldots ,x_{n}\right)  \\
=&e^{-mnt}k_{0}^{(n)}\left( x_{1},\ldots ,x_{n}\right)  +\sigma
\int_{0}^{t}e^{-mn(t-s)}\sum_{i=1}^{n}k_{s}^{(n-1)}\left(
x_{1},\ldots ,\check{x}_{i},\ldots ,x_{n}\right) ds \\
=&\, e^{-mnt}k_{0}^{(n)}\left( x_{1},\ldots ,x_{n}\right)  \\
&+\sigma
e^{-mnt}\int_{0}^{t}e^{mns}\sum_{i=1}^{n}e^{-m(n-1)s}\sum_{\xi
\subset \left\{ x_{1},\ldots ,\check{x}_{i},\ldots ,x_{n}\right\}
}k_{0}^{(\left\vert \xi \right\vert )}\left( \xi \right)
\left(\frac{
\sigma }{m}(e^{mt}-1)\right) ^{n-1-\left\vert \xi \right\vert }ds \\
=&\, e^{-mnt}k_{0}^{(n)}\left( x_{1},\ldots ,x_{n}\right)  \\
& +e^{-mnt}\sum_{\xi \subsetneq \left\{ x_{1},\ldots ,x_{n}\right\}
}\left( n-\left\vert \xi \right\vert \right) k_{0}^{(\left\vert \xi
\right\vert )}\left( \xi \right) \left( \frac{\sigma }{m}\right)
^{n-\left\vert \xi \right\vert }m\int_{0}^{t}\left( e^{ms}-1\right)
^{n-1-\left\vert \xi
\right\vert }e^{ms}ds \\
=&\, e^{-mnt}k_{0}^{(n)}\left( x_{1},\ldots ,x_{n}\right)
+e^{-mnt}\sum_{\xi \subsetneq \left\{ x_{1},\ldots ,x_{n}\right\}
}k_{0}^{(\left\vert \xi \right\vert )}\left( \xi \right) \left(
\frac{
\sigma }{m}(e^{mt}-1)\right) ^{n-\left\vert \xi \right\vert } \\
=&\, e^{-mnt}\sum_{\xi \subset \left\{ x_{1},\ldots ,x_{n}\right\}
}k_{0}^{(\left\vert \xi \right\vert )}\left( \xi \right) \left(
\frac{ \sigma }{m}(e^{mt}-1)\right) ^{n-\left\vert \xi \right\vert
}\\=&\sum_{\xi \subset \left\{ x_{1},\ldots ,x_{n}\right\}
}e^{-m\left\vert \xi \right\vert t}k_{0}^{(\left\vert \xi
\right\vert )}\left( \xi \right) \left(
\frac{\sigma}{m}(1-e^{-mt})\right) ^{n-\left\vert \xi \right\vert }.
\end{align*}
By a mathematical induction principle, the statement is proved.
\end{proof}

\begin{remark}
Note that, by \eqref{freedev-m-dyn}, $k_0(\emptyset)=1$ implies
$k_t(\emptyset)=1$ as well as $k_0>0$ implies
$k_t>0$.
\end{remark}

\begin{proposition}\label{subPsm}
Let $k_0\in\K_C$ and $k_t$ is the solution of \eqref{KmE}. Then $k_t\in\K_{C'}$,
where $C'=\max\{C;\frac{\sigma}{m}\}$. More precisely,
\[
\bigl| k_t(\eta) \bigr| \leq \|k_0\|_{\K_C} \biggl(
\max\Bigl\{C;\frac{\sigma}{m}\Bigr\}\biggr)^{|\eta|}, \quad
\eta\in\Ga_0.
\]
\end{proposition}
\begin{proof}
By \eqref{freedev-m-dyn}, \eqref{RB-norm}, and \eqref{addexp}, one get
\begin{align*}
|k_{t}\left( \eta \right) |\leq &\, e^{-tm|\eta|} \biggl( e_{\lambda
}\left( \frac{ \sigma }{m}(e^{mt}-1)\right) \ast
\bigl(\|k_0\|_{\K_C}e_{\lambda }(C)\bigr)\biggr) \left( \eta \right)
\\=&\, \|k_0\|_{\K_C} e^{-tm|\eta|} e_{\lambda }\left(C+ \sigma
\frac{e^{tm}-1}{m}, \eta\right) \\
 = &\, \|k_0\|_{\K_C} e_{\lambda }\left(Ce^{-tm}+ \sigma \frac{1-e^{-tm}}{m},
\eta\right) \leq\|k_0\|_{\K_C}  e_\la \biggl(
\max\Bigl\{C;\frac{\sigma}{m}\bigr\}, \eta\biggr),
\end{align*}
since
\begin{equation}
Ce^{-tm}+ \sigma \frac{1-e^{-tm}}{m} =
\left(C-\frac{\sigma}{m}\right)e^{-tm} + \frac{\sigma}{m} \leq
\max\left\{C;\frac{\sigma}{m}\right\}. \label{niceest}
\end{equation}
Hence, this dynamics stays so-called \textit{sub-Poissonian} (cf. Remark~\ref{rem:P}
below).
\end{proof}

\begin{remark}\label{rem:P}
Let us stress that if we start in \eqref{KmE} from the Poisson distribution
$\mu_0=\pi_A$ with
$k_{\mu_0}(\eta)=k_0(\eta)=A^{|\eta|}$, $A>0$ then the distribution stays Poissonian during
dynamics:
\begin{equation}\label{Poissondynamics}
k_t(\eta)=\biggl(\Bigl(A-\frac{\sigma}{m}\Bigr)e^{-tm} +
\frac{\sigma}{m} \biggr)^{|\eta|}.
\end{equation}
\end{remark}

\begin{corollary}\label{semigroupexpr}
Let $C\geq\dfrac{\sigma}{m}$. Then for any $k\in\K_C$
\begin{equation}\label{sgcorsm}
\bigl({\hat{T}}^\ast(t)k\bigr)(\eta):=e^{-tm|\eta|} \left( e_{\lambda }\left(
 \frac{\sigma}{m}(e^{tm}-1)\right) \ast k\right) \left( \eta
\right) , \quad \eta\in\Ga_0, \  t>0.
\end{equation}
\end{corollary}

As was noted in \cite{FKK2010}, $\K_\aC\subset D(\hL^*)$ for any
$\a\in(0;1)$. Moreover, by Proposition~\ref{subPsm}, if $k\in\K_\aC$
then $k_t={\hat{T}}^\ast(t)k={\hat{T}} ^\odot(t)k\in\K_{C'}$, where
$C'=\max\{\aC;\frac{\sigma}{m}\}$. Therefore, the following
improvement of the result from \cite{FKK2010} holds.
\begin{proposition}
Let $C>\dfrac{\sigma}{m}$. Then for any
$\a\in\Bigl(\dfrac{\sigma}{mC};1\Bigl)$ the Banach subspace
$\overline{\K_\aC}$ of the Banach space $\K_C$ is ${\hat{T}}
^\odot(t)$-invariant. Here closure is taken in the norm of $\K_C$.
The restriction $\hat{T}^{\odot\a}(t)$ of $\hat{T}^{\odot}(t)$ onto
$\overline{\K_\aC}$ is a contraction $C_0$-semigroup.
\end{proposition}

As a result, we have that for any $C\geq\dfrac{\sigma}{m}$ the
Cauchy problem \eqref{KmE} is solvable on $\K_C$. Moreover, for
$C>\dfrac{\sigma}{m}$ and $\a\in\Bigl(\dfrac{\sigma}{mC};1\Bigl)$
this problem is solvable on $\overline{\K_\aC}$.

At the end let us us find an expression
for the resolvent $R_{z}^{\odot }$
of the generator $\hat{L}^{\odot}$ of the semigroup $\hat{T}^{\odot}(t)$.
\begin{proposition}
For any $z$ with $\mathrm{Re}\,z>0$ there
exists a bounded  operator $R_{z}^{\odot }=(z-\hat{L}^{\odot})^{-1}$ on the space
$\mathring{\K}_C$ such that
for any $k\in\mathring{\K}_C$
\[
\left( R_{z}^{\odot }k\right) \left( \eta \right)  =\frac{1}{m}\sum_{\xi
\subset \eta }\left( \frac{\sigma }{m}\right) ^{\left\vert \xi \right\vert
}B\left( \frac{z}{m}+\left\vert \eta \right\vert -\left\vert \xi \right\vert
,\left\vert \xi \right\vert +1\right) k\left( \eta \setminus \xi \right) ,
\]
where $B(x,y)=\int_{0}^{1}s^{x-1}\left( 1-s\right) ^{y-1 }ds$ is the Euler
beta function.
\end{proposition}

\begin{proof} We have
\begin{align*}
\left( R_{z}^{\odot }k\right) \left( \eta \right)
=&\int_{0}^{\infty }e^{-zt}\,e^{-tm|\eta |}\left( e_{\lambda }\left(
\frac{\sigma }{m}
(e^{tm}-1)\right) \ast k\right) \left( \eta \right) dt \\
=&\sum_{\xi \subset \eta }k\left( \eta \setminus \xi \right)
\int_{0}^{\infty }e^{-\left( z+m|\eta |\right) t}\,\left(
\frac{\sigma }{m} (e^{tm}-1)\right) ^{\left\vert \xi \right\vert
}dt\\=&\sum_{\xi \subset \eta }k\left( \eta \setminus \xi \right)
\left( \frac{\sigma }{m} \right) ^{\left\vert \xi \right\vert
}\int_{0}^{\infty }e^{-\left( z+m|\eta \setminus \xi|\right)
t}\,(1-e^{-tm})^{|\xi|}dt&  .
\end{align*}
Using substitution $s=e^{-tm}$ we obtain
for $\mathrm{Re}\,z>0$
\begin{align}
 \int_{0}^{\infty }e^{-\left( z+m\left\vert \eta \setminus \xi \right\vert
\right) t}\left( 1-e^{-tm}\right) ^{\left\vert \xi \right\vert }dt
=&\,\frac{1}{m}\int_{0}^{1}s^{\frac{z }{m}+\left\vert \eta \setminus \xi
\right\vert -1}\left( 1-s\right) ^{\left\vert \xi \right\vert }ds \notag\\
=&\,\frac{1}{m}B \left( \frac{z}{m}+\left\vert \eta \right\vert -|\xi|,
\left\vert \xi \right\vert +1\right), \label{EulerBeta}
\end{align}
that proves the assertion.
\end{proof}

\subsection{Evolution of measures}

In \cite{FKK2010}, it was shown that dynamics $\hat{T}^{\odot}(t)$
preserves so-called Lenard-positivity property on the subspace
$\overline{D(\hat{L}^\ast)}$. We recall that a measurable function
$k:\Gamma_0\rightarrow{\mathbb{R}}$ is to be called a positive
defined function in the sense of Lenard if for any $G\in
B_{bs}\left( \Gamma _{0}\right) $ such that $KG\geq 0$ the following
inequality holds $\int_{\Gamma _{0}}G\left( \eta \right) k\left(
\eta \right) d\lambda \left( \eta \right) \geq 0$. By
\eqref{likefourier}, any such a function will be positive defined in
the sense of \eqref{ex:posdef} too.

We extend now this preservation of positive-definiteness (in the sense of
\eqref{ex:posdef}) on the whole space $\K_C$.

We start from the following lemma which is seems to be important itself.
\begin{lemma}\label{le:prodonexp}
Let $\mu_0\in\M^1_{\mathrm{fm}}(\Ga)$ and suppose that their
correlation function $k_0=k_{\mu_0}$ exists and belongs to $\K_C$.
Let $f\in L^1 (\X, dx)$ and $0\leq f(x)\leq 1$, $x\in\X$. Then
function $k(\eta)=e_\la (f,\eta) k_0(\eta)$ is a positive definite
in the sense of \eqref{ex:posdef}.
\end{lemma}
\begin{proof}
Using classical measure theory arguments it is enough to proof
\eqref{ex:posdef} for function $G:\Ga_0\rightarrow\C$ of the form
\begin{equation}\label{sumofexp}
G(\eta)=\sum_{i=1}^N b_i e_\la(g_i,\eta), \quad N\in\N, \ b_i\in\C,\
g_i\in C_0(\X\rightarrow\C),
\end{equation}
where $C_0(\X\rightarrow\C)$ is the space of all complex-valued continuous functions on
$\X$ with compact supports.

Note that, by \eqref{RB-norm} and \eqref{LP-exp-mean}, for any $g\in
C_0(\X\rightarrow\C)\subset L^1(\X\rightarrow\C,dx)$
\begin{equation}\label{intest}
\int_{\Ga_0}|e_\la(g,\eta)|k_0(\eta)d\la(\eta)\leq \|k_0\|_{\K_C}
\int_{\Ga_0}e_\la(C|g|,\eta) d\la(\eta)<\infty.
\end{equation}
By \eqref{kexp} and \eqref{Ktransform}, inequality \eqref{intest} implies $\prod_{x\in\ga}(1+|g(x)|)\in L^1(\Ga,d\mu_0)$ for
any $ g\in  C_0(\X)$. Moreover, $\prod_{x\in\ga}(1+|g(x)|)\in\F L^0(\Ga)$,
hence,
\[
\int_{\Ga} \prod_{x\in\ga}(1+g(x)) d\mu_0
(\ga)= \int_{\Ga_\La}
\prod_{x\in\ga_\La}(1+g(x)) d\mu_0^\La (\ga_\La),
\]
where $\La$ is the support of $g$ and the measure $\mu^\La_0$ is the projection of the measure $\mu_0$ onto
$\Ga_\La$.

Let $G$ has the form \eqref{sumofexp}. Then, taking $\La$  equal to union
of the supports of functions $g_i$, $i=1,\ldots,n,$ we obtain
\begin{align*}
&\int_{\Gamma _{0}}\left( G\star \bar{G}\right) \left( \eta \right)
e_{\lambda }\left( f,\eta \right) k_{0}\left( \eta \right) d\lambda
\left(
\eta \right)  \\
=&\sum_{i,j=1}^{N}b_{i}\bar{b}_{j}\int_{\Gamma _{0}}e_{\lambda
}\left( g_{i}+\bar{g}_{j}+g_{i}\bar{g}_{j},\eta \right) e_{\lambda
}\left( f,\eta
\right) k_{0}\left( \eta \right) d\lambda \left( \eta \right)  \\
=&\sum_{i,j=1}^{N}b_{i}\bar{b}_{j}\int_{\Gamma }K\left( e_{\lambda
}\left( fg_{i}+f\bar{g}_{j}+fg_{i}\bar{g}_{j}\right) \right) \left(
\gamma \right)
d\mu _{0}\left( \gamma \right)  \\
=&\sum_{i,j=1}^{N}b_{i}\bar{b}_{j}\int_{\Gamma }\prod\limits_{x\in
\gamma }\Bigl( 1-f\left( x\right) +f\left( x\right) \bigl(
1+g_{i}\left( x\right) \bigr) \bigl( 1+\bar{g}_{j}\left( x\right)
\bigr) \Bigr) d\mu _{0}\left(
\gamma \right)  \\
=&\sum_{i,j=1}^{N}b_{i}\bar{b}_{j}\int_{\Gamma _{\Lambda
}}\prod\limits_{x\in \gamma _{\Lambda }} \Bigl( 1-f\left( x\right)
+f\left( x\right) \bigl( 1+g_{i}\left( x\right) \bigr) \bigl(
1+\bar{g}_{j}\left( x\right) \bigr) \Bigr) d\mu _{0}^{\Lambda
}\left( \gamma _{\Lambda }\right)  \\
=&\sum_{i,j=1}^{N}b_{i}\bar{b}_{j} \int_{\Gamma _{\Lambda
}}\sum_{\eta \subset \gamma _{\Lambda }}e_{\lambda }\left( 1-f,\eta
\right) e_{\lambda }\left( f\left( 1+g_{i}\right) \left(
1+\bar{g}_{j}\right) ,\gamma _{\Lambda }\setminus \eta
\right) d\mu _{0}^{\Lambda }\left( \gamma _{\Lambda }\right)  \\
=&\int_{\Gamma _{\Lambda }}\sum_{\eta \subset \gamma _{\Lambda
}}e_{\lambda }\left( 1-f,\eta \right)
\sum_{i,j=1}^{N}b_{i}\bar{b}_{j}e_{\lambda }\left( 1+g_{i},\gamma
_{\Lambda
}\setminus \eta \right)  \\
&\times e_{\lambda }\left( 1+\bar{g}_{j},\gamma _{\Lambda }\setminus
\eta \right) e_{\lambda }\left( f,\gamma _{\Lambda }\setminus \eta
\right) d\mu
_{0}^{\Lambda }\left( \gamma _{\Lambda }\right)  \\
=&\int_{\Gamma _{\Lambda }}\sum_{\eta \subset \gamma _{\Lambda
}}e_{\lambda }\left( 1-f,\eta \right) \left\vert
\sum_{i=1}^{N}b_{i}e_{\lambda }\left( 1+g_{i},\gamma _{\Lambda
}\setminus \eta \right) \right\vert ^{2}e_{\lambda }\left( f,\gamma
_{\Lambda }\setminus \eta \right) d\mu _{0}^{\Lambda }\left( \gamma
_{\Lambda }\right)\\ \geq & \, 0,
\end{align*}
since $0\leq f(x)\leq 1$, $x\in\X$.
\end{proof}

As we noted before not all elements of $\K_C$ are correlation
functions of some measures.  Next theorem shows that we really have
correlation functions evolutions and, as a result, evolution of
states (measures) on $\bigl(\Ga, \B(\Ga)\bigr)$.

\begin{theorem}
Let $\mu_0\in\M^1_{\mathrm{fm}}(\Ga)$ and $k_0=k_{\mu_0}\in\K_C$,
$C>0$ be the corresponding correlation function on $\Ga_0$. Then for
any $t>0$ the solution $k_t$ of \eqref{KmE} is a correlation function of
a unique measure $\mu_t\in\M^1_{\mathrm{fm}}(\Ga)$.
\end{theorem}
\begin{proof}
By \eqref{freedev-m-dyn}, $k_t$ is positive measurable function and
$k_t(\emptyset)=1$. Proposition~\ref{subPsm} implies sub-Poissonian
bounds for $k_t$. Hence, for apply Theorem~\ref{posdef} we should
check \eqref{ex:posdef} only.

By Lemma~\ref{le:prodonexp}, $e^{-tm|\cdot|}k_0=e_\la(e^{-tm})k_0$
is a positive defined function in the sense of \eqref{ex:posdef}
(cf. \cite[Corollary 3]{KC2006}). Clearly, this function belongs to
$\K_C$. Therefore, by Theorem~\ref{posdef}, there exists a unique
measure from $\M^1_{\mathrm{fm}}(\Ga)$ whose correlation function is
$e_\la(e^{-tm})k_0$.

Next, $e_{\lambda }\left( \frac{\sigma
}{m}(1-e^{-tm})\right)$ is the correlation function of the
Poisson measure with intensity $\frac{\sigma }{m}(1-e^{-tm})$.

By \cite{Fin2010b}, Ruelle convolution of  correlation functions
$e_{\lambda }\left( \frac{\sigma }{m}(1-e^{-tm})\right)$ and
$e_\la(e^{-tm})k_0$ will be positive defined in the sense of
\eqref{ex:posdef} too. Hence, the assertion is followed by
Theorem~\ref{solutionofCP}.
\end{proof}

As it was shown in \cite{Fin2010b}, the $\ast$-convolution of
correlation functions $k_{\mu_1}$ and $k_{\mu_2}$ is the correlation
function of the convolution of measures $\mu_1$ and $\mu_2$, where
by definition $\mu=\mu_1\ast\mu_2$ is the probability measure on
$\bigl( \Ga, \B(\Ga)\bigr)$ such that   for any measurable $ F$ with
$\widetilde{F}\in L^1(\Gamma\times\Gamma,d\mu_1\times d\mu_2),$
where
\begin{equation*}  \label{tildeF}
\widetilde{F}(\gamma_1,\gamma_2)=F(\gamma_1\cup\gamma_2), \quad
\gamma_{1,2}\in\Gamma,
\end{equation*}
the following equality holds
\begin{equation*}  \label{convmeasinf}
\int_\Gamma
F(\gamma)d\mu(\gamma)=\int_{\Gamma}\int_{\Gamma}F(\gamma_1\cup
\gamma_2) \,d\mu_1(\gamma_1)\,d\mu_2(\gamma_2).
\end{equation*}

Let now $\mu_0\in\M^1_{\mathrm{fm}}(\Ga)$ and consider weak evolution equation
for measures:
\begin{equation*}
\frac{\partial}{\partial t} \int_\Ga F(\ga)d\mu_t(\ga)= \int_\Ga (LF)(\ga)d\mu_t(\ga)
\end{equation*}
for any $F\in\F L^0(\Ga)$ provided both parts exist and, of course, $\mu_t
\big|_{t=0}=\mu_0$. Let
$\nu_t\in\M^1_{\mathrm{fm}}(\Ga)$ be solution of a corresponding pure death
evolution equation
\begin{equation*}
\frac{\partial}{\partial t} \int_\Ga F(\ga)d\nu_t(\ga)= m\int_\Ga \sum_{x\in\ga}\bigr(F(\ga\setminus
x)-F(\ga)\bigl)d\nu_t(\ga)
\end{equation*}
with the same initial condition $\nu_t
\big|_{t=0}=\mu_0$. Then, by Theorem~\ref{solutionofCP} for the case $\sigma=0$,
we obtain $k_{\nu_t}(\eta)=e^{-tm|\eta|}k_0(\eta)$. As a result,
\begin{equation}\label{convmeas}
\mu_t= \pi_{z_t}\ast\nu_t,
\end{equation}
where
\[
z_t=\frac{\sigma}{m}(1-e^{-tm}).
\]

\section{Ergodicity}
\subsection{Ergodic properties of correlation functions}
We recall that a measure $\mu_{\mathrm{inv}}\in\M^1_{\mathrm{fm}}(\Ga)$ is
called invariant for the operator $L$ if for any $F\in\F L^0(\Ga)$
\[
\int_\Ga (LF)(\ga) d\mu_{\mathrm{inv}}(\ga) =0.
\]
If $k_{\mathrm{inv}}$ is the corresponding correlation function then for
any $G\in L^0_\mathrm{ls}(\Ga_0)$
\[
\int_{\Ga_0}(\hL^* k_{\mathrm{inv}})(\eta) G(\eta)d\la(\eta)=
\int_{\Ga_0}(\hL G)(\eta) k_{\mathrm{inv}}(\eta)d\la(\eta)=0,
\]
and, therefore, $(\hL^* k_{\mathrm{inv}})(\eta)=0$, $\eta\in\Ga_0$.
As a result, by \eqref{exprLPdual},
\[
m|\eta| k_\mathrm{inv}(\eta)=\sigma
\sum_{x\in \eta }k_\mathrm{inv}\left( \eta \setminus x\right) .
\]
Iterating the last equation, we easily can see that it implies
\begin{equation}
k_{\mathrm{inv}}(\eta) =
\left(\frac{\sigma}{m}\right)^{|\eta|}=e_{\lambda }\left(
\frac{\sigma}{m}, \eta\right) . \label{inv_sm}
\end{equation}
As result, Poisson measure $\pi_{\frac{\sigma}{m}}$
is a unique invariant measure
of our evolution.

Note also that the condition $k_0(\emptyset)=1$ implies that point-wisely
we obtain
\[
k_t(\eta) =
\biggl(\frac{\sigma}{m}\bigl(1-e^{-mt}\bigr)\biggr)^{|\eta|} +
\sum_{\xi\subsetneq\eta} k_0(\eta\setminus\xi)
e^{-mt|\eta\setminus\xi|}\left(\frac{\sigma}{m}\bigl(1-e^{-mt}\bigr)\right)^{|\xi|}
\rightarrow \left(\frac{\sigma}{m}\right)^{|\eta|}
\]
as $t\rightarrow \infty$.
Taking into account \eqref{inv_sm} and
Proposition~\ref{subPsm}, we may expect that our non-equilibrium
dynamics are ergodic in
the space $\mathcal{K}_C$ for big enough $C$.
In the next theorem we explain more exact conditions
for this ergodicity.

\begin{theorem}\label{th:ergodic}
Let $C>\dfrac{\sigma}{m}$, $k_0\in\mathcal{K}_{C}$
and $k_0(\emptyset)=1$. Then
\begin{equation}\label{ergod_ineq}
\bigl\|k_t-k_{\mathrm{inv}}\bigr\|_{\mathcal{K}_C} <  \|k_0-k_\mathrm{inv}\|_{\K_C}  \frac{e^{-mt}}{1-\dfrac{\sigma}{Cm}},\quad t> 0.
\end{equation}
\end{theorem}
\begin{proof}
First of all note that, by \eqref{inv_sm} and Corollary~\ref{subPsm}, for any $C> \dfrac{\sigma}{m}$
\[
\{k_t, t>0;
k_{\mathrm{inv}}\}\subset \mathcal{K}_C.
\]
Next, by \eqref{addexp},
\[
k_{\mathrm{inv}}=e_\la\Bigl(\frac{\sigma}{m}\Bigr)=e_\la\Bigl(\frac{\sigma}{m}(1-e^{-mt})\Bigr)\ast
e_\la \Bigl(\frac{\sigma}{m}e^{-mt}\Bigr).
\]
Therefore,
\begin{equation*}
k_t(\eta) - k_\mathrm{inv} (\eta)=\sum_{\xi\subsetneq\eta} \left(k_0(\eta\setminus\xi)-\Bigl(\frac{\sigma}{m}\Bigr)^{|\eta\setminus\xi|}
\right)e^{-mt|\eta\setminus\xi|}\left(\frac{\sigma}{m}\bigl(1-e^{-mt}\bigr)\right)^{|\xi|}
\end{equation*}and one can estimate

\begin{align}\label{e:specest}
&C^{-|\eta|} \Bigl| k_t(\eta) - k_\mathrm{inv} (\eta) \Bigr|\\ \leq
& \, C^{-|\eta|}  \sum_{\xi\subsetneq\eta}
\biggl|k_0(\eta\setminus\xi)-\Bigl(\frac{\sigma}{m}\Bigr)^{|\eta\setminus\xi|}
\biggr|\,
e^{-mt|\eta\setminus\xi|}\left(\frac{\sigma}{m}\bigl(1-e^{-mt}\bigr)\right)^{|\xi|}
 \notag\\ \leq & \, \|k_0-k_\mathrm{inv}\|_{\K_C}C^{-|\eta|} \sum_{\xi\subsetneq\eta}C^{|\eta\setminus\xi|}
e^{-mt|\eta\setminus\xi|}\left(\frac{\sigma}{m}\bigl(1-e^{-mt}\bigr)\right)^{|\xi|}\nonumber\\
=& \, \|k_0- k_\mathrm{inv}
\|_{\K_C}\left[\left(e^{-mt}+\frac{\sigma}{Cm}\bigl(1-e^{-mt}\bigr)\right)^{|\eta|}
-\Bigl(\frac{\sigma}{Cm}\bigl(1-e^{-mt}\bigr)\Bigr)^{|\eta|}\right].\nonumber
\end{align}
Let us recall, that
$e^{-mt}+\frac{\sigma}{Cm}\bigl(1-e^{-mt}\bigr)<
1$, $t>0$.

To find uniform, by $|\eta|$, estimate
for the r.h.s. of \eqref{e:specest} let us
consider for any fixed $0<a<b<1$, $n\in\N$ the difference
\begin{align*}
b^n-a^n =& (b-a)\sum_{j=0}^{n-1} a^j b^{n-1-j}<
(b-a)\sum_{j=0}^{n-1} a^j \\=&(b-a)\frac{1-a^n}{1-a}
<\frac{b-a}{1-a}.
\end{align*}

As result, using  \eqref{e:specest}, obvious estimate
$\dfrac{\sigma}{Cm}\bigl(1-e^{-mt}\bigr)<\dfrac{\sigma}{Cm}$, $t>0$, and the fact that $\dfrac{1}{1-a}$ is a strictly increasing function of $a\in(0;1)$
we obtain \eqref{ergod_ineq}.
\end{proof}

\begin{remark}
Note that if we consider corresponding general result from
\cite{FKK2010} in the zero-potential case and for $m=1$ we obtain
more weaker inequality
\[
\|k_t -k_\mathrm{inv}\|_{\K_C}\leq e^{-(1-\nu)t} \|k_0 -k_\mathrm{inv}\|_{\K_C}, \quad1>\nu>\frac{\sigma}{C}.
\]
\end{remark}

Let $\La\in\B_c(\X)$ and denote the projection of the measure
$\mu_t$, $t\geq0$ on $\Ga_\La$ by $\mu_t^\La$. Then, in the same
notations, $\mu _{inv}^{\Lambda }=\pi^\La_{{\sigma}/{m}}$.
\begin{corollary}
Let $C>\dfrac{\sigma}{m}$ and $A=\Bigl(1-\dfrac{\sigma}{mC} \Bigr)^{-1}$. Then for any $t>0$
\begin{equation}\label{proj-ergodic}
\left\Vert\frac{d\mu _{t}^{\Lambda }}{d\lambda } - \frac{d\mu
_{\mathrm{inv}}^{\Lambda }}{d\lambda }\right\Vert_{\K_C}\leq Ae^{-tm}\exp
\left\{ C\left\vert \Lambda \right\vert \right\}.
\end{equation}
\end{corollary}
\begin{proof}
Since, clearly, $\int_{\Ga_\La} 2^{|\eta|}
k_t(\eta)d\la(\eta)<+\infty$, $t\geq0$ then (see, e.g.,
\cite{KK2002})
\begin{equation}\label{proj-sm}
\frac{d\mu _{t}^{\Lambda }}{d\lambda }\left( \eta \right)
=\int_{\Gamma _{\Lambda }}\left( -1\right) ^{\left\vert \xi
\right\vert }k_{t}\left( \eta \cup \xi \right) d\lambda \left( \xi
\right), \quad \eta\in\Ga_\La.
\end{equation}
Hence, by Theorem~\ref{th:ergodic}, for $C>\frac{\sigma }{m},t>0$, we have
\begin{align*}
&\left\vert \frac{d\mu _{t}^{\Lambda }}{d\lambda }\left( \eta
\right) -
\frac{d\mu _{\mathrm{inv}}^{\Lambda }}{d\lambda }\left( \eta \right) \right\vert  \\
\leq &\int_{\Gamma _{\Lambda }}\left\vert k_{t}\left( \eta \cup \xi
\right) -k_{\mathrm{inv}}\left( \eta \cup \xi \right) \right\vert
d\lambda \left( \xi \right)
\\
=&\int_{\Gamma _{\Lambda }}\frac{\left\vert k_{t}\left( \eta \cup
\xi \right) -k_{\mathrm{inv}}\left( \eta \cup \xi \right)
\right\vert }{C^{\left\vert \eta \cup \xi \right\vert
}}C^{\left\vert \eta \cup \xi \right\vert
}d\lambda \left( \xi \right)  \\
\leq &\left\Vert k_{t}-k_{\mathrm{inv}}\right\Vert
_{\mathcal{K}_{C}}C^{\left\vert
\eta \right\vert }\exp \left\{ C\left\vert \Lambda \right\vert \right\}  \\
\leq & AC^{\left\vert \eta \right\vert }e^{-tm}\exp \left\{
C\left\vert \Lambda \right\vert \right\} ,
\end{align*}
that proves the assertion.
\end{proof}

For any $\eta\in\Ga_0$, $y\in\X$, $t\geq 0$ we define
\begin{equation}\label{alaUrs}
v_t(\eta,y):=k_t(\eta\cup y)-k_t(\eta)k_t(y).
\end{equation}
Clearly, Remark~\ref{rem:P} implies that
if $k_0(\eta)=A^{|\eta|}$, $A>0$ then $v_t(\eta,y)\equiv
0$.

Our dynamics at moment $t$ is said to
be satisfied {\em the decay of correlation principle}
if
\begin{equation}\label{cordecay}
\lim_{|y|\rightarrow\infty} v_t(\eta,y)=0,\quad
\eta\in\Ga_0.
\end{equation}

Next theorem shows preserving the decay of correlation
principle during our dynamics.
\begin{theorem}\label{noneq-fact}
Let $C>\dfrac{\sigma}{m}$, $k_{0}( \emptyset
) =1$  and let
\[
a(y) := \bigl\|v_0(\cdot,y)\bigr\|_{\mathcal{K}_C}\in[0;\,\infty),
\quad y\in\X.
\]
Then
\[
\bigl\|v_t(\cdot,y)\bigr\|_{\mathcal{K}_C} \leq a(y) e^{-tm},\quad
y\in\X.
\]
\end{theorem}
\begin{proof}
Let $y\in\X$ be fixed. Consider the mapping
\begin{equation}\label{defder}
(D_y G)(\eta):=G(\eta\cup \{y\}).
\end{equation}
By direct computations,
we obtain from \eqref{defast} that $D_y$ is satisfied chain rule:
\begin{equation}\label{derprop}
D_y(G_1*G_2)=(D_y G_1)*G_2+G_1*(D_y G_2).
\end{equation}

Therefore,
\begin{align*}
D_y \left( e_{\lambda }\left( \frac{\sigma }{m}\left(
e^{tm}-1\right) \right) \ast k_{0}\right) =&  \frac{\sigma }
{m}\left( e^{tm}-1\right) \left( e_{\lambda }\left( \frac{\sigma
}{m}\left( e^{tm}-1\right)
\right) \ast k_{0}\right) \\
&\quad + \left( e_{\lambda }\left( \frac{\sigma }{m}\left(
e^{tm}-1\right) \right) \ast k_{0}(\cdot\cup y)\right).
\end{align*}
Hence, using equality
\begin{equation}\label{first-sm}
k_t(y)=e^{-tm}\left(k_{0}\left(  y \right) +\frac{\sigma }{m} \left( e^{tm}-1\right)\right),
\end{equation}
we obtain
\begin{align*}
v_{t}\left( \eta ,y\right)  =&\, e^{-tm(\left\vert \eta
\right\vert+1)}D_{y}\left( e_{\lambda }\left( \frac{\sigma
}{m}\left( e^{tm}-1\right)
\right) \ast k_{0}\right) \left( \eta  \right)  \\
& -e^{-tm\left\vert \eta \right\vert }\left( e_{\lambda }\left(
\frac{\sigma }{m}\left( e^{tm}-1\right) \right) \ast k_{0}\right)
\left( \eta \right)
k_{t}\left( y\right)  \\
=&\, e^{-tm(\left\vert \eta \right\vert+1)} \frac{\sigma }{m}\left(
e^{tm}-1\right) \left( e_{\lambda }\left( \frac{\sigma }{m}\left(
e^{tm}-1\right) \right) \ast k_{0}\right) \left( \eta \right)
\\& +e^{-tm(\left\vert \eta \right\vert+1)}\left( e_{\lambda
}\left( \frac{\sigma }{m}\left( e^{tm}-1\right) \right) \ast
k_{0}\left( \cdot \cup y\right) \right) \left( \eta \right)
\\
&-e^{-tm(\left\vert \eta \right\vert+1) }\left( e_{\lambda }\left(
\frac{ \sigma }{m}\left( e^{tm}-1\right) \right) \ast k_{0}\right)
\left( \eta \right) \left[ k_{0}\left( y \right) +\frac{\sigma }{m}
\left( e^{tm}-1\right) \right]  \\
=&\, e^{-tm(\left\vert \eta \right\vert+1)}\left( e_{\lambda }\left(
\frac{\sigma }{m}\left( e^{tm}-1\right) \right) \ast k_{0}\left(
\cdot \cup y\right) \right) \left( \eta \right)
\\
& -e^{-tm(\left\vert \eta \right\vert+1) } k_{0}\left( y \right)
\left( e_{\lambda }\left( \frac{\sigma }{m}\left( e^{tm}-1\right)
\right) \ast k_{0}\right) \left( \eta \right)
\\
=&\, e^{-tm(\left\vert \eta \right\vert+1)}\sum_{\xi \subset \eta
}\left( \frac{\sigma }{m}\left( e^{tm}-1\right) \right) ^{\left\vert
\eta \setminus \xi \right\vert }v_{0}(\xi,y).
\end{align*}
Therefore, for any $\eta\in\Ga_0$ one has
\begin{align*}
C^{-|\eta|}\bigl| v_t(\eta,y) \bigr|\leq & \,
C^{-|\eta|}e^{-tm(\left\vert \eta \right\vert+1)}\sum_{\xi \subset
\eta }\left( \frac{\sigma }{m}\left( e^{tm}-1\right) \right)
^{\left\vert \eta \setminus
\xi \right\vert }C^{|\xi|}C^{-|\xi|}\bigl |v_{0}(\xi,y)\bigr|\\
\leq & \,a(y)C^{-|\eta|}e^{-tm(\left\vert \eta
\right\vert+1)}\sum_{\xi \subset \eta }\left( \frac{\sigma
}{m}\left( e^{tm}-1\right) \right) ^{\left\vert \eta \setminus
\xi \right\vert }C^{|\xi|}\\
=&\,a(y)C^{-|\eta|}e^{-tm(\left\vert \eta \right\vert+1)} \left(
C+\frac{
\sigma }{m}\left( e^{tm}-1\right) \right) ^{\left\vert \eta \right\vert }\\
=&\,a\left( y\right) e^{-tm}\left( e^{-tm}+\frac{\sigma }{Cm}\left(
1-e^{-tm}\right) \right) ^{\left\vert \eta \right\vert }\leq a\left(
y\right) e^{-tm}.
\end{align*}
The statement is proved.
\end{proof}
\begin{remark}
From the proof of the Theorem~\ref{noneq-fact}
one can see that if \eqref{cordecay} holds for $t=0$ then
it holds for any $t>0$ as well.
\end{remark}

\begin{remark}
More traditional object for studying decay of correlation principle
is the so-called Ursell functions (or truncated correlation
functions). We recall (see \cite{Fin2010b} and references therein)
that  the\ function $u_t:\Ga_0\rightarrow\R$ is called Ursell
function for $k_t$ if
\[
k_t=\exp^\ast u_t:=\sum_{n=0}^\infty\frac{1}{n!}u_{t}^{\ast n}, \quad u^{*0}:=1^\ast.
\]
The condition $k_t(\emptyset)=1$ guarantees existence of $u_t$ with
$u_t(\emptyset)=0$ (see e.g. \cite{Fin2010b} for details). Then, by
\eqref{freedev-m-dyn-alt} and \cite{Fin2010b}, we obtain that $u_t$
is equal to sum of Ursell functions, corresponding to correlation
functions of measures $\pi_{z_t}$ and $\nu_t$ from \eqref{convmeas}.
It's easy to see that the Ursell function corresponding to the
Poisson measure $\pi_{z_t}$ is equal to
$\chi_{\{|\eta|=1\}}\dfrac{\sigma}{m}(1-e^{-tm})$. Next,
\begin{align*}
e^{-tm|\eta|}k_0(\eta)=&\,e^{-tm|\eta|}\sum_{n=1}^\infty \frac{1}{n!}u^\ast_0(\eta)\\
=&\,e^{-tm|\eta|}\sum_{n=1}^\infty
\frac{1}{n!}\sum_{\substack{\eta_1\sqcup\ldots\sqcup\eta_n=\eta\\\eta_i\neq\emptyset,\
1\leq i \leq n}}u(\eta_1)\ldots u(\eta_n)\\=&\sum_{n=1}^\infty
\frac{1}{n!}\sum_{\substack{\eta_1\sqcup\ldots\sqcup\eta_n=\eta\\\eta_i\neq\emptyset,\
1\leq i \leq n}}e^{-tm|\eta_{1}|}u(\eta_1)\ldots
e^{-tm|\eta_{n}|}u(\eta_n)
\\=&\,\exp^*\bigl(e^{-tm|\cdot|}u_0\bigr)(\eta).
\end{align*}
Therefore,
\begin{equation*}
u_t(\eta)=e^{-tm|\eta|}u_0(\eta)+\chi_{\{|\eta|=1\}}\dfrac{\sigma}{m}(1-e^{-tm}).
\end{equation*}
In particular, if for any $n\geq 2$ the symmetric function
$u_0^{(n)}$ is integrable by $j$ variables ($1\leq j\leq n-1$) then
$u_t^{(n)}$ has this property too.
\end{remark}

\subsection{Evolution of Bogolyubov functional}
Let $\mu\in\M^1_\mathrm{fm}(\Ga)$ such that for any $\theta\in L^1(\X,dx)$
the following so-called Bogolyubov functional there exists:
\begin{equation}\label{BF}
B_\mu(\theta):=\int_\Ga \prod_{x\in\ga}\bigl(1+\theta(x)\bigr)d\mu(\ga).
\end{equation}
By \eqref{Ktransform} and \eqref{kexp}, we have an another representation
\begin{equation}\label{BF0}
B_\mu(\theta)=\int_{\Ga_0} e_\la(\theta,\eta)k_\mu(\eta)d\la(\eta).
\end{equation}
In particular, if there exists $C>0$ such that $k_\mu(\eta)\leq \mathrm{const}\cdot
C^{|\eta|}$, $\eta\in\Ga_0$ then, by \eqref{LP-exp-mean} and \eqref{kexp},
the r.h.s. of \eqref{BF0} as well as \eqref{BF}
are finite.

\begin{proposition}
Let $C>\dfrac{\sigma}{m}$, $k_0\in\mathcal{K}_{C}$
and $k_0(\emptyset)=1$. Let $B_t(\theta):=B_{\mu_t}(\theta)$, $B_\mathrm{inv}(\theta):=B_{\mu_\mathrm{inv}}(\theta)$.
Then
\begin{equation*}
\bigl\vert B_t(\theta) -B_\mathrm{inv}(\theta)\bigr\vert\leq e^{-mt} \|k_0-k_\mathrm{inv}\|_{\K_C}  \frac{\exp\bigl\{ C\Vert\theta\Vert_{L^1}\bigr\}}{1-\dfrac{\sigma}{Cm}}.
\end{equation*}
\end{proposition}
\begin{proof}
First of all let us note that, by Proposition~\ref{subPsm}, $B_t$ exists.
Then, by Theorem~\ref{th:ergodic}, we have
\begin{align*}
\bigl\vert B_t(\theta)-B_\mathrm{inv}(\theta)\bigr\vert
=&\Biggl\vert \int_{\Ga_0} e_\la(\theta,\eta)k_t (\eta)d\la(\eta) -  \int_{\Ga_0} e_\la(\theta,\eta)k_\mathrm{inv}(\eta) d\la(\eta)\Biggr\vert\\
\leq & \int_{\Ga_0} e_\la(|\theta|,\eta)\bigl\vert k_t (\eta) -
k_\mathrm{inv}(\eta)
\bigr\vert d\la(\eta)\\
\leq & \, \Vert k_t - k_\mathrm{inv}\Vert_{\K_C} \int_{\Ga_0} e_\la(|\theta|,\eta) C^{|\eta|} d\la(\eta)\\
\leq & \, \|k_0-k_\mathrm{inv}\|_{\K_C}
\frac{e^{-mt}}{1-\dfrac{\sigma}{Cm}} \exp\bigl\{
C\Vert\theta\Vert_{L^1}\bigr\}. \qedhere
\end{align*}
\end{proof}
\begin{remark}
Note that, by \eqref{freedev-m-dyn}, we have
\begin{align}
B_{t}\left( \theta \right)  =&\int_{\Gamma _{0}}e_{\lambda }\left( \theta
,\eta \right) k_{t}\left( \eta \right) d\lambda \left( \eta \right)  \nonumber\\
=&\int_{\Gamma _{0}}e_{\lambda }\left( \theta ,\eta \right) e^{-tm\left\vert \eta \right\vert }\int_{\Gamma
_{0}}e_{\lambda }\left( \theta ,\xi
\right) e_{\lambda }\left( \frac{\sigma }{m}\left( 1-e^{-tm}\right) ,\xi
\right) d\lambda \left( \xi \right) k_{0}\left( \eta \right) d\lambda \left(
\eta \right)  \nonumber\\
=&\int_{\Gamma _{0}}e_{\lambda }\left( \theta ,\xi \right) e_{\lambda
}\left( \frac{\sigma }{m}\left( 1-e^{-tm}\right) ,\xi \right) d\lambda
\left( \xi \right) \int_{\Gamma _{0}}e^{-tm\left\vert \eta \right\vert
}e_{\lambda }\left( \theta ,\eta \right) k_{0}\left( \eta \right) d\lambda
\left( \eta \right)  \nonumber\\
=&\exp \left\{ \frac{\sigma }{m}\left( 1-e^{-tm}\right) \left\langle \theta
\right\rangle \right\} B_{0}\left( e^{-tm}\theta \right),\label{BFexpr}
\end{align}
that corresponds to \eqref{convmeas}.

Since, by \eqref{BF0} and \eqref{LP-exp-mean}, $B_\mathrm{inv}(\theta)=\exp\Bigl\{\dfrac{\sigma}{m}\langle\theta\rangle\Bigr\}$,
we obtain from \eqref{BFexpr}
\begin{align}
&B_t(\theta) - B_\mathrm{inv}(\theta)\nonumber\\=&\exp \left\{ \frac{\sigma }{m}\left( 1-e^{-tm}\right) \left\langle \theta
\right\rangle \right\} \Bigl( B_{0}\left( e^{-tm}\theta \right) - \exp \left\{ \frac{\sigma }{m} e^{-tm} \left\langle \theta
\right\rangle \right\}\Bigr) \nonumber\\=&
\exp \left\{ \frac{\sigma }{m}\left( 1-e^{-tm}\right) \left\langle \theta
\right\rangle \right\} \int_{\Gamma _{0}}e^{-tm\left\vert \eta \right\vert
}e_{\lambda }\left( \theta ,\eta \right)\bigl( k_{0}\left( \eta \right) - k_\mathrm{inv}(\eta)
\bigr) d\lambda
\left( \eta \right). \label{diffBF}
\end{align}
As a result, if, e.g., $k_0(\eta)\leq \Bigl(\dfrac{\sigma}{m}\Bigr)^{|\eta|}=k_\mathrm{inv}(\eta)$,
$\eta\in\Ga_0$ and $k_0(\emptyset)=1$, then for any $0\leq\theta\in L^1(\X,dx)$ one has
\[
0\leq B_\mathrm{inv}(\theta)-B_t(\theta) \leq e^{-tm} \exp \left\{ \frac{\sigma }{m}\left( 1-e^{-tm}\right) \left\langle \theta
\right\rangle \right\} \bigl( B_\mathrm{inv}(\theta)-B_0(\theta)\bigr).
\]
\end{remark}

One can consider now the state space where the evolution $B_0(\theta)\mapsto
B_t(\theta)$ lives. Let $E=L^1(\X,dx)$. We recall (see, e.g., \cite{KKO2006}
and references therein) that a functional $A:E\rightarrow\C$ is called entire
on $E$ whenever $A$ is locally bounded and for all $\theta_0, \theta\in E$ the mapping $\C\ni z\mapsto A(\theta_0+z\theta)\in\C$ is entire. For any
$\alpha>0  $ we consider a Banach space $E^{(\a)}$ of entire functionals
on $E$ with norm
\[
\Vert A\Vert_{\alpha}:=\sup_{\theta\in E}\bigl( |A(\theta)| e^{-\alpha \Vert\theta\Vert_ E}\bigr)<\infty.
\]
Then for any $\a\geq \dfrac{\sigma}{m}$ we
have
\begin{align*}
&\|B_t(\theta)\|_\a=\sup_{\theta\in E}\left(\exp\left\{\dfrac{\sigma}{m}(1-e^{-tm})\langle\theta\rangle\right\}
|B_0(e^{-tm}\theta)|\exp\{-\a\|\theta\|_E\}\right)\\
\leq&\|B_0(\theta)\|_\a \sup_{\theta\in E}\left(\exp\left\{\dfrac{\sigma}{m}(1-e^{-tm})\|\theta\|_E\right\}
\exp\left\{\a (e^{-tm}-1)\|\theta\|_E\right\}\right)\\
=&\|B_0(\theta)\|_\a \sup_{\theta\in E}\left(\exp\left\{\left(\dfrac{\sigma}{m}-\a\right)(1-e^{-tm})\|\theta\|_E\right\}
\right)\leq \|B_0(\theta)\|_\a.
\end{align*}
Therefore, the evolution $B_0(\theta)\mapsto
B_t(\theta)$ preserves balls in $E^{(\a)}$.

\section{Evolution on $\L_C$}

We recall now without a proof the partial case of the well-known
lemma (cf., \cite{KMZ2004}).
\begin{lemma}\label{Minlos}
For any measurable function
$H:\Ga_0\times\Ga_0\times\Ga_0\rightarrow\R$
\begin{equation}\label{minlosid}
\int_{\Gamma _{0}}\sum_{\xi \subset \eta }H\left( \xi ,\eta
\setminus \xi ,\eta \right) d\lambda \left( \eta \right)
=\int_{\Gamma _{0}}\int_{\Gamma _{0}}H\left( \xi ,\eta ,\eta \cup
\xi \right) d\lambda \left(  \xi \right) d\lambda \left(  \eta
\right)
\end{equation}
if only both sides of the equality make sense.
\end{lemma}

Next statements present explicit form for the semigroup on $\L_C$
and resolvent of its generator and show mean-ergodic properties of
this semigroup (see, e.g., \cite{EN2000} for a terminology).

\begin{proposition}
Let $C\geq\dfrac{\sigma}{m}$. Then for any  $G\in\L_C$
\begin{equation}\label{Lsm-sg}
\bigl(\hat{T}(t)G\bigr)\left( \eta \right) =e^{-tm\left\vert \eta
\right\vert }\int_{\Gamma _{0}}G\left( \eta \cup \xi \right)
e_{\lambda }\left( \frac{\sigma }{m} \left( 1-e^{-tm}\right) ,\xi
\right) d\lambda \left( \xi \right).
\end{equation}
Moreover, for any $z\in\C$ with $\mathrm{Re}\,z>0$ there exist bounded resolvent
operator $R_z=(\hat{L}-z)^{-1}$ and for any $G\in\L_C$
\begin{equation}\label{resolvent}
\left( R_{z}G\right) \left( \eta \right)
=\frac{1}{m}\int_{\Gamma _{0}}G\left( \eta \cup \xi \right) \left( \frac{\sigma }{m}
\right) ^{\left\vert \xi \right\vert } B \left( \frac{z}{m}+\left\vert
\eta
\right\vert ,  \left\vert \xi \right\vert +1\right) d\lambda \left( \xi \right),
\end{equation}
where $B(x,y)=\int_{0}^{1}s^{x-1}\left( 1-s\right) ^{y-1 }ds$ is the Euler
beta function.
\end{proposition}
\begin{proof}
Let $C\geq\dfrac{\sigma}{m}$ and $G\in\L_C$. Then, $\hat{T}(t)G\in\L_C$ and for any $k\in\K_C$, by Corollary~\ref{semigroupexpr} and  Lemma~\ref{Minlos}, one has
\begin{align*}
&\int_{\Gamma _{0}}\bigl(\hat{T}(t)G\bigr)\left( \eta \right) k\left( \eta \right)
d\lambda \left( \eta \right)  =\int_{\Ga_0}
G(\eta) \bigl(\hat{T}^*(t)k\bigr)\left( \eta \right)d\la(\eta)\\
=& \int_{\Gamma _{0}}G_{}\left( \eta \right) e^{-tm\left\vert \eta
\right\vert }\left( e_{\lambda }\left( \frac{\sigma }{m }\left(
e^{tm}-1\right) \right) \ast k\right) \left( \eta \right)
d\lambda \left( \eta \right)  \\
=&\int_{\Gamma _{0}}\int_{\Gamma _{0}}G\left( \eta \cup \xi \right)
e^{-tm\left\vert \eta \right\vert }e^{-tm\left\vert \xi \right\vert
}e_{\lambda }\left( \frac{\sigma }{m}\left( e^{tm}-1\right) ,\xi
\right) k\left( \eta \right) d\lambda \left( \xi \right) d\lambda
\left( \eta \right),
\end{align*}
that implies \eqref{Lsm-sg}.

Since $\hat{T}(t)$ is a $C_0$-semigroup with generator $\bigl( \hL,
D(\hL)\bigr)$ then  for any $z\in\C$ with $\mathrm{Re}\,z>0$
\[
R_{z}=\int_{0}^{\infty }e^{-zt}U\left( t\right) dt.
\]
Then, by direct computation,
\begin{align*}
\left( R_{z}G\right) \left( \eta \right)  =&\int_{0}^{\infty
}e^{-zt}e^{-tm\left\vert \eta \right\vert }\int_{\Gamma _{0}}G\left(
\eta \cup \xi \right) e_{\lambda }\left( \frac{\sigma }{m}\left(
1-e^{-tm}\right)
,\xi \right) d\lambda \left( \xi \right) dt \\
=&\int_{\Gamma _{0}}G\left( \eta \cup \xi \right) \left(
\frac{\sigma }{m} \right) ^{\left\vert \xi \right\vert
}\int_{0}^{\infty }e^{-\left( z+m\left\vert \eta \right\vert \right)
t}\left( 1-e^{-tm}\right) ^{\left\vert \xi \right\vert }dtd\lambda
\left( \xi \right) ,
\end{align*}
and the assertion follows from \eqref{EulerBeta}.
\end{proof}

\begin{theorem}
Let $C\geq \max \left( \dfrac{\sigma }{m};\,1\right) $ and $G\in
\mathcal{L} _{2C}$ then
\begin{equation*}
\frac{1}{t}\int_{0}^{t}\hat{T}\left( s\right) Gds\rightarrow
{\chi}_{\Gamma ^{(0)}}\cdot \int_{\Gamma _{0}}G\left( \xi \right)
k_{\mathrm{inv}}\left( \xi \right) d\lambda \left( \xi \right)
\end{equation*}
as $t\rightarrow \infty $ in $\mathcal{L}_{C}$.
\end{theorem}

\begin{proof} Using equality
\begin{equation*}
{1}_{\Gamma ^{(0)}}(\eta )\cdot \int_{\Gamma _{0}}G\left( \xi
\right) k_{ \mathrm{inv}}\left( \xi \right) d\lambda \left( \xi
\right) =\int_{\Gamma _{0}}G\left( \xi \cup \eta \right) \left(
\frac{\sigma }{m}\right) ^{\left\vert \xi \right\vert }0^{\left\vert
\eta \right\vert }d\lambda \left( \xi \right)
\end{equation*}
we have
\begin{align*}
& \left\Vert \frac{1}{t}\int_{0}^{t}\hat{T}\left( s\right)
Gds-{\chi}_{\Gamma ^{(0)}}\cdot \int_{\Gamma _{0}}G\left( \xi
\right) \left( \frac{\sigma }{m} \right) ^{\left\vert \xi
\right\vert }d\lambda \left( \xi \right)
\right\Vert _{\mathcal{L}_{C}} \\
\leq & \int_{\Gamma _{0}}\int_{\Gamma _{0}}\left\vert G\left( \eta
\cup \xi \right) \right\vert \left( \frac{\sigma }{m}\right)
^{\left\vert \xi \right\vert }\left\vert
\frac{1}{t}\int_{0}^{t}e^{-sm\left\vert \eta \right\vert }\left(
1-e^{-sm}\right) ^{\left\vert \xi \right\vert }ds-0^{\left\vert \eta
\right\vert }\right\vert C^{\left\vert \eta \right\vert }d\lambda
\left( \xi \right) d\lambda \left( \eta \right) .
\end{align*}
We have
\begin{align*}
&\frac{1}{t}\int_{0}^{t}e^{-sm\left\vert \eta \right\vert }\left(
1-e^{-sm}\right) ^{\left\vert \xi \right\vert }ds \\
=&\sum_{j=0}^{\left\vert \xi \right\vert }\binom{\left\vert \xi
\right\vert }{j}\left( -1\right)
^{j}\frac{1}{t}\int_{0}^{t}e^{-sm\left( \left\vert \eta
\right\vert +j\right) }ds \\
=&\begin{cases} \displaystyle\sum_{j=0}^{\left\vert \xi \right\vert
}\binom{\left\vert \xi \right\vert }{j}\left( -1\right)
^{j}\frac{1}{t}\frac{1-e^{-tm\left( \left\vert \eta \right\vert
+j\right) }}{m\left( \left\vert \eta \right\vert +j\right) }, \quad
|\eta|\neq 0 \\[5mm]
\displaystyle 1+ \sum_{j=1}^{\left\vert \xi \right\vert
}\binom{\left\vert \xi \right\vert }{j}\left( -1\right)
^{j}\frac{1}{t}\frac{1-e^{-tmj}}{mj }, \quad |\eta|= 0
\end{cases}
\end{align*}
Therefore, for any $\xi ,\eta \in \Gamma _{0}$
\begin{equation*}
\frac{1}{t}\int_{0}^{t}e^{-sm\left\vert \eta \right\vert }\left(
1-e^{-sm}\right) ^{\left\vert \xi \right\vert }ds\rightarrow
0^{\left\vert \eta \right\vert },\quad t\rightarrow \infty .
\end{equation*}

Using trivial estimate $\left\vert
\dfrac{1}{t}\displaystyle\int_{0}^{t}e^{-sm\left\vert \eta
\right\vert }\left( 1-e^{-sm}\right) ^{\left\vert \xi \right\vert
}ds-0^{\left\vert \eta \right\vert }\right\vert \leq 1$ we obtain
the assertion by the dominated convergence theorem since
\begin{align*}
& \int_{\Gamma _{0}}\int_{\Gamma _{0}}\left\vert G\left( \eta \cup
\xi \right) \right\vert \left( \frac{\sigma }{m}\right) ^{\left\vert
\xi \right\vert }C^{\left\vert \eta \right\vert }d\lambda \left( \xi
\right)
d\lambda \left( \eta \right)  \\
=& \int_{\Gamma _{0}}\left\vert G\left( \eta \right) \right\vert
\left( 1+ \frac{\sigma }{Cm}\right) ^{\left\vert \eta \right\vert
}C^{\left\vert \eta
\right\vert }d\lambda \left( \eta \right)  \\
\leq & \int_{\Gamma _{0}}\left\vert G\left( \eta \right) \right\vert
2^{\left\vert \eta \right\vert }C^{\left\vert \eta \right\vert
}d\lambda \left( \eta \right) =\left\Vert G\right\Vert _{2C}<+\infty
.
\end{align*}
The statement is proved.
\end{proof}

\end{document}